\documentclass[12pt,leqno]{amsart}
\usepackage{amsmath,epsfig,graphicx,color, float}
\usepackage{subcaption}
\usepackage{relsize}
\usepackage{comment}
\usepackage[colorlinks, linkcolor=blue,citecolor=blue]{hyperref}
\usepackage{pdfpages}
\usepackage{graphicx}
\usepackage{mwe}
\usepackage{longtable} 
\usepackage{algorithm}
\usepackage{algorithmicx}
\usepackage{algpseudocode}
\usepackage{natbib}
\usepackage{array}
\usepackage{ragged2e}
\usepackage{booktabs, makecell, tabularx}

\usepackage[margin=1.3in]{geometry}

\usepackage{booktabs} 
\usepackage{tabularx} 
\usepackage{makecell} 

\numberwithin{table}{section}
\numberwithin{figure}{section}
\numberwithin{equation}{section}

\definecolor{darkblue}{rgb}{.2, 0.2,.8}
\definecolor{darkgreen}{rgb}{0,0.5,0.3}
\definecolor{darkred}{rgb}{.8, .1,.1}

\newcommand{\bfx}{\vect{x}}
\newcommand{\bfX}{\vect{X}}

\newcommand{\bfalp}{\vect{\alpha}}

\newcommand{\bftheta}{\vect{\theta}}
\newcommand{\bfbeta}{\vect{\beta}}

\newcommand{\bfT}{\mat{T}}
\newcommand{\bft}{\vect{t}}
\newcommand{\bfe}{\vect{e}}

\newcommand{\bfp}{\vect{\pi}}

\newcommand{\0}{\mat{0}}

\renewcommand{\P }{{\mathrm{pr}}}

\newcommand{\eqd}{\stackrel{d}{=}}

\newtheorem{lemma}{Lemma}[section]

\newtheorem{proposition}[lemma]{Proposition}
\newtheorem{definition}[lemma]{Definition}

\newtheorem{condition}[lemma]{Condition}
\newtheorem{remark}{Remark}[section]

\usepackage{bm}
\newcommand{\vect}[1]{\pmb{#1}}
\newcommand{\mat}[1]{\boldsymbol{\bm #1}}

\DeclareMathOperator*{\argmax}{arg\,max}

\usepackage{marginnote}

\allowdisplaybreaks

\begin{document}
\bibliographystyle{apalike}
\title[Cure models: from mixture to matrix distributions]{Cure models: from mixture to matrix distributions}

\author[M. Bladt]{Martin Bladt}
\address{Department of Mathematical Sciences, 
University of Copenhagen,
DK-2100 Copenhagen, 
Denmark}
\email{martinbladt@math.ku.dk}

\author[J. Yslas]{Jorge Yslas}
\address{Institute for Financial and Actuarial Mathematics,
University of Liverpool,
L69 7ZL Liverpool,
UK}
\email{jorge.yslas-altamirano@liverpool.ac.uk}

\begin{abstract}
Cure rate models address survival data in which a proportion of individuals will never experience the event of interest. Existing parametric approaches are predominantly based on finite mixtures, which impose restrictive assumptions on both the cure mechanism and the distribution of susceptible event times. A cure model based on phase-type distributions is introduced, leveraging their latent Markov jump process representation to allow immunity to occur either at baseline or dynamically during follow-up. This structure yields a flexible and interpretable formulation of long-term survival while encompassing classical mixture cure models as special cases. A unified regression framework is developed for covariate effects on both the cure rate and the susceptible survival distribution, and the proposed model class is dense, reducing the impact of parametric misspecification. Estimation is performed via expectation–maximization algorithms, accompanied by an automatic model selection strategy. Simulation studies and a real-data example demonstrate the practical advantages of the approach.
\end{abstract}
\maketitle

\section{Introduction}
In survival analysis, the standard assumption that all subjects eventually experience the event of interest is often violated. In many clinical studies, particularly in oncology, a proportion of patients achieve long-term remission and never experience disease recurrence, regardless of how long the follow-up period extends. These individuals are referred to as {\em cured}, and their presence is typically signaled by a leveling-off of the Kaplan-Meier survival curve at the tail of the follow-up period. Standard survival models, which assume all subjects will eventually experience the event, are inadequate in these scenarios and can lead to biased estimates and loss of information within the data~\citep[cf., e.g.,][]{amico2018cure}. 

To accommodate this phenomenon, cure rate models have been developed. The seminal work of~\citet{boag1949maximum} introduced the mixture cure model, later refined by~\citet{berkson1952survival}, which conceptualizes the population as comprising two sub-populations: susceptible individuals who eventually experience the event, and cured individuals who do not. Under this framework, the survival function $S(t)$ is modeled as $$ S(t) = (1 - p) + p \cdot S_u(t), $$ where $p$ is the proportion of susceptible individuals (so that $1-p$ is the cure rate), and $S_u(t)$ is the survival function for the susceptible sub-population. This formulation naturally separates the model into an {\em incidence} component (the probability of being susceptible) and a {\em latency} component (the time-to-event distribution for susceptibles). Several extensions of this framework have been proposed; for instance,~\citet{farewell1982use} extended it by incorporating covariates into the incidence component via logistic regression, while~\citet{Kuk1992, SyTaylor2000} developed semiparametric approaches using proportional hazards models for the latency. We refer to~\citet{MallerZhou1996, peng2021cure} for a comprehensive account of mixture cure rate models with various survival functions.  

Alternative formulations to the mixture cure rate model have been introduced in the literature. For instance, the promotion time cure model was introduced by~\citet{Yakovlev1996} and further developed by~\citet{ChenIbrahimSinha1999}. This approach is motivated by biological considerations, in which the event time is conceptualized as the first arrival among a random number of competing risks, with the number of risks following a Poisson distribution. Frailty-based approaches provide yet another perspective, where cure arises naturally when the frailty distribution has an atom at zero~\citep{price2001modelling}. Despite their conceptual differences, these approaches can often be expressed within a common mixture-type representation separating incidence and latency.

While a wide range of parametric and semiparametric cure models has been proposed, practical applications frequently reveal complex, multi-modal, or nonmonotone hazard structures that are difficult to capture with standard distributions. Moreover, most of these frameworks assume that cure status is determined at baseline (time zero), which may not always be appropriate when cure can effectively occur at any time during follow-up. These challenges are compounded in regression settings, where both the cure probability and the latency distribution may depend on covariates. As a result, there is a growing interest in highly flexible yet tractable cure-rate models that can accommodate complex survival dynamics without sacrificing interpretability or computational feasibility. 

In this paper, we introduce a novel and flexible class of cure rate models based on phase-type distributions. Defined as the distribution of the absorption time of a finite-state continuous-time Markov jump process, phase-type distributions are dense in the class of distributions on $[0,\infty)$ and enjoy attractive analytical and computational properties, including closed-form expressions for key functionals and efficient maximum-likelihood estimation via expectation-maximization (EM) algorithms; see~\cite{Bladt2017} for a comprehensive treatment. Their application to survival analysis was first explored by~\cite{aalen1995phase}, who demonstrated their flexibility in modeling various hazard function shapes. More recent applications in biostatistical and actuarial contexts include mortality modeling~\citep{albrecher2022survival} and frailty models~\citep{yslas2025phase}.
 
 Our proposed specification extends the phase-type framework by defining a Markov jump process on a state space that includes two distinct absorbing states: one representing the event of interest (e.g., death or disease recurrence) and another representing immunity (cure). This structure, driven by a hidden Markov jump process, provides a richer and more dynamic interpretation of the curing mechanism. Under this formulation, a subject can either be immune from the outset or can transition to an immune state over time, generalizing the static mixture model, which assumes immunity is determined at time zero. 
 
The contribution of this paper extends existing cure-rate modeling in several important directions. We develop a flexible parametric framework that naturally accommodates a cure fraction through an absorbing state representing immunity, and we derive explicit expressions for both the cure rate and the latency distribution in terms of the underlying Markov process parameters. Building on this, we introduce a regression framework based on the Mixture-of-Experts (MoE) specification~\citep{bladt2022moe}, allowing covariates to simultaneously influence the cure probability and the survival distribution of susceptibles. We further establish the denseness of this class of MoE models within the family of cure regression models, ensuring that any well-behaved cure model can be approximated arbitrarily well and thereby reducing the risk of parametric misspecification. Finally, we develop effective estimation procedures based on the EM algorithm, along with practical tools for model selection and goodness-of-fit assessment using Cox-Snell-type residuals.

The remainder of this paper is structured as follows. Section~\ref{sec2} introduces the cure rate phase-type model, derives the explicit expression for the cure rate, and details the estimation procedure. Section~\ref{sec3} discusses the regression framework. Section~\ref{sec4} addresses the practical issue of model selection and goodness-of-fit. Section~\ref{sec5} and Section~\ref{sec6} present a simulation study and an application to a leukemia dataset, respectively. Finally, Section~\ref{sec:conc} concludes the paper.

\section{Absorption times of Markov jump-processes}\label{sec2}

\subsection{General specification}\label{subsec1}

We consider a time-homogeneous Mar\-kov pure-jump process $ ( \mathcal{X}_t )_{t \geq 0}$ on a finite state space which we may label as $\{1,2, \dots, r, r+1\}$, where states $1,\dots,r-1$ are transient, state $r$ represents immunity, and is absorbing, and state $r+1$ represents the event of interest, death, and is also absorbing.  Then, the transition probabilities
\begin{align*}
p_{kl}(s,t)=\P(\mathcal{X}_t=l\mid \mathcal{X}_s=k)\,,\quad 1\le k,l\le r+1 \,, \quad 0\le s< t \,,
\end{align*}
can be written in matrix form as 
\begin{align*}
\mat{P}(s,t)=\exp(\mat{\Lambda}(t-s)),	
\end{align*}
where $\mat{\Lambda}$ is called the intensity matrix, having negative diagonal elements (for non-absorbing states) and non-negative off-diagonal elements such that the rows sum to zero, and the exponential of a square matrix $\mat{A}$ is defined by the power series 
\begin{align*}
	\exp(\mat{A}) = \sum_{j= 0}^{\infty} \frac{\mat{A}^{j}}{j!} \, .
\end{align*}

Our assumptions imply that we may write
\begin{align*}
	\mat{\Lambda}=\left( \begin{array}{cc}
		\bfT &  \bft \\
		\0 & 0
	\end{array} \right)\in\mathbb{R}^{(r+1)\times(r+1)}\,,
\end{align*}
where $\bfT $ is an $r \times r$ sub-intensity matrix, $\bft$ is an $r$-dimensional column vector providing the exit rates to the absorbing state describing the event of interest, and $\0$ is a $r$-dimensional row vector of zeroes.  In fact, the last row and last element of $\bfT=(t_{kl})_{k,l=1,\dots,r}$ and $\bft=(t_1,\dots,t_r)^{\mathsf{T}}$, respectively, are also zero.  Moreover, since the rows of the intensity matrix sum to zero, the relationship $\bft=- \bfT \, \bfe$ is seen to hold, where $\bfe$ denotes the $r$-dimensional column vector of ones.  

Though both $r$ and $r+1$ are absorbing states, we are interested in them for different reasons.  On the one hand, we would like to estimate the proportion of subjects who eventually end up being absorbed in $r$, which, following cure rate terminology, we may refer to as \textit{insufficiency}.  On the other hand, we would like to estimate the distribution of the time-to-event, that is, the first hitting time to the absorbing state $r+1$, occurring only to subjects who have no immunity, and which we refer to as \textit{latency}.  

\begin{definition}
Let $\mathcal{X}_0\sim\bfp = (\pi_1, \dots, \pi_r)$ be the initial distribution on $\{1,\dots,r\}$.  We say that
$$\tau = \inf \{ t >  0 : \mathcal{X}_t = r+1 \}\,,$$ 
follows a (defective) phase-type distribution and we write  $\tau \sim \mbox{PH}(\bfp, \bfT)$.  
\end{definition}

The above definition makes it clear how to deal with latency, which is by analyzing the distribution of $\tau$. The density $f_\tau$ and distribution function $F_\tau$ of $\tau \sim  \mbox{PH}(\bfp, \bfT)$ are explicit in terms of functions of matrices, as follows, for $t\ge0$:
\begin{align}
 f_\tau(t) &= \, \vect{\pi}\exp \left(  \mat{T} t \right)\vect{t}\,, \nonumber\\
 F_\tau(t) &=  1- \vect{\pi}\exp \left( \mat{T} t \right)\vect{e}\, \,.\label{ph_cdf}
\end{align}

\subsection{The cure rate}\label{subsec_curerate}
Insufficiency may be retrieved by considering
\begin{align*}
p:=\P(\tau<\infty)=\lim_{t\to\infty}F_\tau(t)\in(0,1] \,,
\end{align*}
estimating the susceptible fraction of the population, and thus $1-p$ is the \textit{cure rate}, or immune fraction. If the dynamics of the underlying Markov process do not allow for transitions from one of the states $1,\cdots,r-1$ to $r$, then it is clear that $1-p=\pi_r$. However, in the case where immunity may be obtained during the lifetime of the process $\mathcal{X}$, we have the following result. 

We use the notation $\bfT_{a:b,c:d}$, $a,b,c,d\in\mathbb{N}$ to denote the sub-matrix of $\bfT$ running from $a\ge1$ to $b\ge a$ in the first coordinate and from $c\ge1$ to $d\ge c$ in the second coordinate. For a $q$-dimensional vector $\vect{s}$, we denote by $\vect{s}_{a:b}$ the sub-vector with indices from $a\ge1$ to $b\ge a$. We also define $\Delta(\vect{s})$ as the $q \times q$ diagonal matrix with entries those of $\vect{s}$ on its main diagonal. Throughout, $\0$ denotes a vector or matrix of zeros of appropriate dimension.

\begin{proposition}\label{prop:cure_rate}
The cure rate of $\tau \sim \mbox{PH}(\bfp, \bfT)$ is given by
\begin{align}
1-p=\P(\tau=\infty)=\bfp_{1:(r-1)}[-\bfT_{1:(r-1),1:(r-1)}]^{-1}\bfT_{1:(r-1),r} + \pi_r \,. \label{cure_rate_ph}
\end{align}
Moreover, the survival function $ S_\tau  = 1 -  F_\tau $ of $\tau$ admits the representation
\begin{align*}
 S_\tau(t) &= (1 - p) + p \cdot S_u(t) \,,
\end{align*}
where $S_u$ denotes the survival function of a phase-type distribution with vector of initial probabilities
\begin{align}\label{eq:latency}
\frac{\bfp_{1:(r-1)}[-\bfT_{1:(r-1),1:(r-1)}]^{-1} [-\bfT_{1:(r-1),1:(r-1)} -  \Delta(\bfT_{1:(r-1),r}) ]}{\bfp_{1:(r-1) }[-\bfT_{1:(r-1),1:(r-1)}]^{-1} \bft_{1:(r-1)} }
\end{align}
and sub-intensity matrix $\bfT_{1:(r-1),1:(r-1)}$. 

\end{proposition}
\begin{proof}
By Theorem 1 of~\cite{VanLoan1978}, we may represent the matrix exponential of a block matrix as
\begin{align}
\exp \left( \mat{T} t \right)= \left( \begin{array}{cc}
		\exp(\bfT_{1:(r-1),1:(r-1)}t) &  \int_{0}^t \exp({\bfT_{1:(r-1),1:(r-1)}(t-z)}) \bfT_{1:(r-1),r} \mathrm{d}z\\
		\0 & 1
	\end{array} \right),
	\label{eq:vanloan}
\end{align}
from which we deduce that
\begin{align*}
\lim_{t\to\infty}\exp \left( \mat{T} t \right)= \left( \begin{array}{cc}
		\0 &  [-\bfT_{1:(r-1),1:(r-1)}]^{-1} \bfT_{1:(r-1),r}\\
		\0 & 1
	\end{array} \right).
\end{align*}
The first result then follows directly by matrix multiplication using equation~\eqref{ph_cdf}. The second part follows straightforwardly by substituting~\eqref{eq:vanloan} into~\eqref{ph_cdf} and rearranging for the term given in~\eqref{cure_rate_ph}. 
\end{proof}

\subsection{Censoring and estimation}\label{subsec_cens}
We work in the survival framework where we do not actually observe $\tau$, but rather $\xi=\min\{\tau,\zeta\}$, $\delta=1\{\tau\le\zeta\}$, where $\zeta$ is a censoring mechanism, independent of $\tau$.  Assume that we have replicate observations $(\xi_i,\delta_i)_{i=1,\dots,n}$.  Thus, for estimation purposes, one should solve the following optimization problem arising from maximum likelihood estimation:
\begin{align}
\argmax\limits_{(\bfp, \bfT)} \sum_{i = 1}^{n} \delta_i \log(f_\tau(\xi_i;\,\bfp, \bfT))+(1-\delta_i) \log (S_\tau(\xi_i;\,\bfp, \bfT)) \quad \mbox{s.t.}\quad \bfT_{r,1:r}=\0. 
\end{align}
Then the cure rate is obtained as a by-product, using~\eqref{cure_rate_ph}.  In other words, estimation of $p$ is not numerically more complex than estimating the underlying general phase-type distribution; in fact, under our framework, both are naturally estimated at the same time.  

Direct optimization, however, is somewhat intensive using gradient methods.  Instead, we employ a slightly adapted version of the expectation-maximization (EM) algorithm from~\cite{AsmussenNermanOlsson1996,Olsson1996}, using the latent path of the underlying Markov jump process.  Indeed, care must be taken regarding setting the row of the immune state $r$ to zeros at the initial parameters, but the convenient thing about the algorithm is that it naturally respects zeros.  In other words, the constraint is respected at any iteration if the initial value does so as well.  The algorithm is provided for completeness in Algorithm~\ref{algo1}, and is implemented using, for example, the \texttt{matrixdist}  R package~\citep{bladt2025matrixdist}.
\begin{algorithm}[!htp]
\caption{Cure rate phase-type estimation}\label{algo1}
\begin{algorithmic}[1]
\Require Observations $(\xi_i,\delta_i)_{i=1,\dots,n}$, and initial parameters $(\bfp,\bfT)$ with $\bfT_{r,1:r}=\0$. 
\Ensure Estimates for the triple $(p,\bfp,\bfT)$.  
\State Set $\vect{\eta}_i=\delta_i\bft+(1-\delta_i)\bfe$.
\State Initialize with $\ell=1.$
\While{$\ell \le N_{\text{steps}}$}
	\item Calculate
\begin{align*}
	& A_k
	=\sum_{i=1}^{n}  \frac{  \pi_k \, \bfe^{\top}_{k} \exp( \bfT \xi_i ) \vect{\eta}_i }{\bfp \exp( \bfT \xi_i ) \vect{\eta}_i }  ,\quad B_k
	 = \sum_{i=1}^{n} \frac{ \int^{\xi_i}_{0}  \bfe^{\top}_{k} \exp({ \bfT(\xi_i-u)}) \vect{\eta}_i \bfp \exp( \bfT u )\bfe_{k}   du }{\bfp \exp( \bfT \xi_i ) \vect{\eta}_i} , \\[3mm]
	& C_{kl}
	 = \sum_{i=1}^{n}  t_{kl} \frac{ \int^{\xi_i}_{0}  \bfe^{\top}_{l} \exp( \bfT(\xi_i-u)) \vect{\eta}_i \bfp \exp( \bfT u )\bfe_{k}   du }{\bfp \exp( \bfT \xi_i ) \vect{\eta}_i},  \quad  D_k
	 = \sum_{i:\, \delta_i=1}  t_{k}   \frac{\bfp \exp( \bfT \xi_i )\bfe_{k}  }{\bfp \exp( \bfT \xi_i ) \bft} .
\end{align*}
	\item Let
	\begin{align*}
		&\pi_{k} =\frac{A_k}{n}  
		\,, \quad
		t_{kl} = \frac { C_{kl} }{B_k }
		\,, \quad t_{k} = \frac {D_k } { B_k}
		\,, \quad
		t_{kk} = -\sum_{l \neq k} t_{kl} -t_{k} \,.
	\end{align*}
\State $\ell = \ell+1$
\EndWhile
\State Set $p=1-\bfp_{1:(r-1)}[-\bfT_{1:(r-1),1:(r-1)}]^{-1}\bfT_{1:(r-1),r} - \pi_r.$
\end{algorithmic}
\end{algorithm}

\subsection{Inhomogeneity}

Distributions following~\eqref{ph_cdf} are dense in the sense of weak convergence on all distribution functions with non-negative support~\citep[cf.][]{Bladt2017}.  
This implies that any target distribution can be approximated arbitrarily closely by a phase-type model provided the order $r$ is allowed to increase sufficiently. 
In practice, this means that a given dataset may sometimes require  a rather large dimension $r$ to obtain an adequate fit.  A possible workaround is to consider inhomogeneous phase-type distributions, which introduce a time transformation, allowing for more flexible tail behaviors without increasing the dimension $r$.

Originally introduced by~\cite{albrecher2019inhomogeneous}, inhomogeneous phase-type distributions generalize phase-type distributions by allowing the Markov process $( \mathcal{X}_t )_{t \geq 0}$ to be time-inhomogeneous. This means the process is governed by a time-dependent intensity matrix
\begin{align*}
	\mat{\Lambda}(t)=\left( \begin{array}{cc}
		\bfT(t) &  \bft(t) \\
		\0 & 0
	\end{array} \right)\in\mathbb{R}^{(r+1)\times(r+1)}\,,
\end{align*}
where $ \bft(t) = -\bfT(t)\bfe$. While this general formulation is mathematically tractable, it is often too complex for direct applications as its functionals are given in terms of product integrals. A particularly useful subclass arises when $\bfT(t)=\lambda(t)\bfT$, in which case we write $\tau \sim \mbox{IPH}(\bfp, \bfT, \lambda)$. 
 Here, $\mat{T}$ is a sub-intensity matrix, and $\lambda(t)$ is a non-negative, real-valued function, known as the intensity function, that can be interpreted as a rate that speeds up or slows down the transitions of the process over time. 
  Members of this subclass also allow for a transformation interpretation of a phase-type model.
If $\tau \sim \mbox{IPH}(\bfp, \bfT, \lambda)$, then there exists a function $g(t)$ such that $\tau \eqd g(Y)$, where $Y \sim \mbox{PH}(\bfp, \bfT)$. The transformation $g(t)$ relates to the intensity function $\lambda(t)$ via $g^{-1}(t) = \int_{0}^{t} \lambda(s) ds$.

The density and survival functions for $\tau \sim \mbox{IPH}(\bfp, \bfT, \lambda)$ are given by
\begin{align*}
 f_\tau(t) &= \lambda(t) \, \vect{\pi}\exp \left( g^{-1}(t) \mat{T}  \right)\vect{t} \,, \\
 S_\tau(t) &= \vect{\pi}\exp \left( g^{-1}(t) \mat{T}  \right)\vect{e}\,.
\end{align*}
Regarding the cure rate, it is clear from the form of the survival function and the proof of Proposition~\ref{cure_rate_ph} that the cure rate is given by~\eqref{cure_rate_ph}.

This framework provides significant flexibility. The asymptotic tail behavior of the distribution is determined by the choice of $g(t)$ (or equivalently of $\lambda(t)$), allowing the model to capture a wide variety of tail shapes beyond the purely exponential tails of standard phase-type distributions. By selecting different forms for $g(t)$, several classical distributions commonly used in survival analysis, such as the Weibull, Lognormal, Log-logistic, and Gompertz distributions, can be extended to have matrix-valued parameters, creating a rich and versatile family of distributions for cure rate modeling. The specific transformations leading to these matrix extensions can be found in Table~\ref{tab:inho}, and we refer to~\cite{albrecher2019inhomogeneous, albrecher2022fitting, albrecher2022survival} for further details and applications of these models.

\begin{table}[!htb]
\centering
\caption{Transformations and survival functions of some inhomogeneous phase-type models.}
\begin{tabularx}{\textwidth}{lcccc}
\toprule
& \makecell{${g(t)}$} & \makecell{{Parameters}\\ {Domain}} & \makecell{${S_\tau(t)}$}  \\
\midrule
Matrix-Weibull  & 
    $t^{1/\beta}$ & 
    $\beta>0$ & 
    $\bfp \exp\left( \bfT t^{\beta} \right)\bfe$ \\[4mm]

Matrix-Lognormal & 
    $\exp(t^{1/\gamma})-1$ & 
    $\gamma > 1$ & 
    $\bfp \exp\left(\bfT \left(\log \left(t+1\right)\right)^{\gamma}\right) \bfe $ \\[4mm]

Matrix-Log-logistic  & 
    $\gamma (\exp(t)-1)^{1/\theta}$ & 
    $\gamma,\theta >0$ & 
    $\bfp \exp\left( \bfT \log\left( \left( t/\gamma \right)^{\theta} + 1 \right) \right) \bfe$ \\[4mm]

Matrix-Gompertz  & 
    $\log( \beta t + 1 ) / \beta$ & 
    $\beta > 0$ & 
    $\bfp \exp \left( \bfT (\exp\left(\beta t\right) -1) / \beta \right) \bfe$ \\[2mm]
    
\bottomrule
\end{tabularx}
\label{tab:inho}
\end{table}

Estimation of the inhomogeneous phase-type model parameters is straightforward using a modified version of Algorithm~\ref{algo1}. The approach, introduced in~\cite{albrecher2022fitting}, embeds the update for the transformation's parameters within the EM framework.
The main idea is to assume that $g(\cdot)$  is a parametric function depending on some parameter $\bftheta$, let us say $g(\cdot; \bftheta)$, and  proceed iteratively as follows. In each step, holding the transformation parameter $\bftheta$ fixed, an inverse transformation $g^{-1}(t; \bftheta)$ is applied to the observed data. This yields a new dataset that can be treated as a sample from a standard phase-type distribution, allowing the parameters $\vect{\pi}$ and $\mat{T}$ to be updated using the standard EM steps from Algorithm~\ref{algo1}. Following this, an additional step is performed in which the parameter of the transformation $g(t;\bftheta)$ is updated by directly maximizing the incomplete-data loglikelihood, holding $\vect{\pi}$ and $\mat{T}$ fixed. As shown in~\cite{albrecher2022fitting}, these iterative steps guarantee that the loglikelihood increases at each iteration.

\section{Regression on the cure rate and survival shape}\label{sec3}
To understand how individual characteristics affect both the probability of being cured and the survival pattern of the non-cured population, we introduce a regression framework based on the concept of Mixture-of-Experts (MoE), as originally introduced in~\cite{bladt2022moe}. The construction principle consists of letting the covariates determine the initial state of the underlying Markov process, while the sub-intensity matrix $\bfT$ remains common across individuals. This is called the phase-type MoE regression model. The different initial states can be seen as ``experts", each corresponding to a specific phase-type survival distribution. The covariates then assign a particular mixture of these experts to each individual. By determining the process starting point, covariates shape the entire survival trajectory, including the cure rate and the latency distribution.

\subsection{Phase-type MoE}
For an individual with covariate column vector $\bfX = \bfx$, we allow the initial distribution of the Markov process to depend on $\bfx$ by specifying a covariate-dependent initial probability vector $\vect{\pi}(\bfx) = (\pi_1(\bfx), \dots, \pi_r(\bfx))$. Specifically, we model $\vect{\pi}(\bfx)$ using a multinomial logistic (softmax) regression,
\begin{align*}
\pi_k(\bfx) = \P(\mathcal{X}_0 = k \mid \bfx) = \frac{\exp(\vect{\beta}_k \bfx)}{\sum_{j=1}^{r} \exp(\vect{\beta}_j \bfx)}\,, \quad k=1, \dots, r\,,
\end{align*}
where $\vect{\beta}_k$ is a $h$-dimensional vector of regression coefficients associated with the $k$-th initial state. 
The resulting covariate-dependent survival function is then given by
\begin{align*}
S_\tau(t \mid \bfx) = \bfp(\bfx) \exp(\mat{T}t)\vect{e},
\end{align*}
yielding the following expression for the covariate-dependent cure rate
\begin{align*}
1 - p(\bfx) = \bfp_{1:(r-1)}(\bfx)[-\bfT_{1:(r-1),1:(r-1)}]^{-1}\bfT_{1:(r-1),r} + \pi_r(\bfx).
\end{align*}
This expression, together with the equivalent of~\eqref{eq:latency}, shows explicitly how covariates, by influencing the initial probabilities $\pi_k(\bfx)$, simultaneously affect the proportion of individuals who are cured and the survival dynamics of those who are not.

Estimation of the phase-type MoE regression model can be performed using a modified EM algorithm based on Algorithm~\ref{algo1}. While the E-step remains conceptually the same, the M-step is adapted to handle the regression coefficients $\vect{\beta} = (\vect{\beta}_k)_{k = 1, \dots, r}$. This is achieved by fitting a weighted multinomial logistic regression, where the weights correspond to the expected number of times the process starts in each state for each individual, as computed in the E-step. For a detailed derivation of this procedure, we refer to~\cite{bladt2022moe}.
For completeness, we provide the general algorithm that incorporates both the regression framework and the inhomogeneity transformation in Algorithm~\ref{algo2}.

\begin{algorithm}[!htp]
\caption{Cure rate inhomogeneous phase-type MoE estimation}\label{algo2}
\begin{algorithmic}[1]
\Require Observations $(\xi_i,\delta_i, \bfx_i)_{i=1,\dots,n}$, and initial parameters $(\vect{\beta}, \bfT, \bftheta)$ with $\bfT_{r,1:r}=\0$. 
\Ensure Estimates for the triple $(\vect{\beta},\bfT, \bftheta)$ and covariant-dependent susceptible fractions $p(\bfx_i)$.  
\State Set $\vect{\eta}_i=\delta_i\bft+(1-\delta_i)\bfe$ and 
\begin{align*}
		\pi_k(\bfx_i) =  \frac{\exp(\vect{\beta}_k \bfx_i)}{\sum_{j=1}^{r} \exp(\vect{\beta}_j \bfx_i)}, \quad i = 1,\dots, n, \; k=1, \dots, r.
	\end{align*}
\State Initialize with $\ell=1.$
\While{$\ell \le N_{\text{steps}}$}

	\item Transform the data into $y_i = g^{-1}(\xi_i;\bftheta)$.
	\item Calculate
\begin{align*}
	& A_k^{i}
	=  \frac{  \pi_k(\bfx_i)\, \bfe^{\top}_{k} \exp( \bfT y_i ) \vect{\eta}_i }{\bfp(\bfx_i) \exp( \bfT y_i ) \vect{\eta}_i }  , \quad B_k
	 = \sum_{i=1}^{n} \frac{ \int^{y_i}_{0}  \bfe^{\top}_{k} \exp({ \bfT(y_i-u)}) \vect{\eta}_i \bfp(\bfx_i) \exp( \bfT u )\bfe_{k}   du }{\bfp(\bfx_i) \exp( \bfT y_i ) \vect{\eta}_i} , \\[3mm]
	& C_{kl}
	 = \sum_{i=1}^{n}  t_{kl} \frac{ \int^{y_i}_{0}  \bfe^{\top}_{l} \exp( \bfT(y_i-u)) \vect{\eta}_i \bfp(\bfx_i) \exp( \bfT u )\bfe_{k}   du }{\bfp(\bfx_i) \exp( \bfT y_i ) \vect{\eta}_i},  \quad  D_k
	 = \sum_{i:\, \delta_i=1}  t_{k}   \frac{\bfp(\bfx_i) \exp( \bfT y_i )\bfe_{k}  }{\bfp(\bfx_i) \exp( \bfT y_i ) \bft} .
\end{align*}
	\item Let
	\begin{align*}
		t_{kl} = \frac { C_{kl} }{B_k }
		\,, \quad t_{k} = \frac {D_k } { B_k}
		\,, \quad
		t_{kk} = -\sum_{l \neq k} t_{kl} -t_{k} \,,
	\end{align*}
	\begin{align*}
	\bfbeta=\argmax_{\bfbeta\in{\mathbb{R}}^{(r\times h)}}\sum_{i=1}^{n}\sum_{k = 1 }^{r} A_k^{i} \log(\pi_k(\vect{x}_i) )\,.
\end{align*} 
\item Let
		\begin{align*}
	 \bftheta
	&= \argmax_{ \bftheta } \sum_{i=1}^{n} \log \left( \lambda(\xi_i ;  \bftheta )^{\delta_i} {\vect{\pi}}(\vect{x}_i) \exp\left({ g^{-1}(\xi_i;\bftheta){\bfT} }\right) \vect{\eta}_i \right) \,.
	\end{align*}
	
\State $\ell = \ell+1$
\EndWhile
\State Set $p(\bfx_i)=1-\bfp_{1:(r-1)}(\bfx_i)[-\bfT_{1:(r-1),1:(r-1)}]^{-1}\bfT_{1:(r-1),r} - \pi_r(\bfx_i)$.
\end{algorithmic}
\end{algorithm}

\subsection{Denseness on cure models}\label{subsec2}
Let $\tau$ be the time to the event of interest and let $Z$ be a binary cure indicator, where $Z=0$ if an individual is cured and $Z=1$ if they are susceptible.

\begin{definition}
	Let  $\mathcal{A}\subset\mathbb{R}^h$ be a given covariate space. A {cure regression model} is then defined as the collection of all conditional laws
\[
\Big\{
  \P(\tau\in \cdot \mid \boldsymbol{X}=\boldsymbol{x})
  : \boldsymbol{x}\in\mathcal{A}
\Big\}.
\]
Furthermore, a sequence of cure regression models is said to converge weakly (respectively, uniformly weakly) to a given model if the conditional distributions converge weakly for each $\boldsymbol{x}\in\mathcal{A}$ (respectively, uniformly weakly in $\boldsymbol{x}\in\mathcal{A}$).

\end{definition}
We are particularly interested in cure regression models that admit the general mixture representation
\begin{align}\label{eq:gen_cure}
	S(t\mid \bfx) = \P(\tau > t \mid \bfX=\bfx) = (1-p(\bfx)) + p(\bfx)S_u(t \mid \bfx)\,,
\end{align}
where $p(\bfx) = \P(Z=1 \mid \bfX=\bfx)$ is the probability of being susceptible, and $S_u(t\mid\bfx) = \P(\tau>t \mid Z=1, \bfX=\bfx)$ is the survival function for the susceptible subpopulation. Note that we have assumed, without loss of generality, that the incidence and latency depend on the same set of covariates, as we can always augment the covariates employed in each component to contain the same. 

\begin{remark} \rm 
The general mixture representation~\eqref{eq:gen_cure} encompasses several cure rate models commonly used in survival analysis. For instance, the classical mixture cure proportional hazards model arises when the latency component $S_u(t\mid\bfx)$ follows a proportional hazards structure and the incidence $p(\bfx)$ is modeled separately through a logistic regression. Frailty-based cure models also fit into~\eqref{eq:gen_cure} when the frailty distribution contains an atom at zero, in which case cured individuals correspond to those with zero frailty and the susceptible survival is given by the conditional frailty model. Similarly, promotion time cure models, where the event time is defined as the minimum of a random number of latent competing risks, admit a mixture representation with $p(\bfx)$ governing the probability of at least one active risk and $S_u(t\mid\bfx)$ describing the conditional survival given susceptibility.
\end{remark}

The key theoretical advantage of our proposed phase-type MoE model is its denseness, which means it is flexible enough to approximate any well-behaved cure regression model. To formalize this, we require some regularity conditions.

\begin{definition}
A feature space $\mathcal{A}$ is said to be regular if it is of the form $\mathcal{A}=\{1\}\times [a,b]^{h-1},$ $a,b\in\mathbb{R}$, that is, covariate vectors include an intercept and the remaining components are contained in a compact hypercube.
\end{definition}

\begin{condition}
A cure rate regression model is said to satisfy the tightness and Lipschitz conditions on $\mathcal{A}$ if the family of latency distributions $\{\P(\tau \in \cdot\,|\,\bfX=\vect{x}, Z = 1)\}_{\vect{x}\in\mathcal{A}}$ is tight, and for each $t\geq 0$, the function $\vect{x}\mapsto\P(\tau \le t \,|\,\bfX=\vect{x}, Z = 1)$ is Lipschitz continuous in $\mathcal{A}$.
\end{condition}

\begin{proposition}\label{prop:dense}
Let a cure regression model that admits the representation~\eqref{eq:gen_cure} satisfy the tightness and Lipschitz conditions on a regular $\mathcal{A}$.  Then there exists a sequence of phase-type MoE regression models that converges (uniformly) weakly to it.  
\end{proposition}

\begin{proof}
One can examine the proof of Theorem 3.3 in~\cite{fung2019class} to obtain that a modified version of logit-weighted reduced MoE (LRMoE) models of the form
\begin{align*}
	F(\xi ; \bfx) = \left( 1 - \frac{\exp(\bfalp_{(r + 1)} \bfx)}{\sum_{k = 1}^{r + 1} \exp(\bfalp_{k} \bfx)} \right) \left[ \sum_{j = 1}^{r} \left( \frac{\exp(\bfalp_{j} \bfx)}{\sum_{k = 1}^{r} \exp(\bfalp_{k} \bfx)} \right) F(\xi ; \bftheta_j) \right]
\end{align*}
with experts $F( \cdot ;\bftheta_j)$ satisfying Property 3 of Proposition 3.1 in~\cite{fung2019class} posses the corresponding denseness property. When the experts are chosen to be Erlang distributions, one obtains a particular instance of the phase-type MoE model. Given that Erlang distributions meet the aforementioned property, the result follows.  
\end{proof}

\section{Model selection and goodness-of-fit}\label{sec4}

A critical aspect of implementing a phase-type cure model is selecting the dimension $r$. The choice of $r$ involves a trade-off between model fit and complexity. A larger $r$ increases model flexibility and often yields a better or equal-quality fit, but it also increases the number of parameters, the computational cost, and the risk of overfitting.

Traditional model selection criteria, such as the Akaike Information Criterion (AIC) and the Bayesian Information Criterion (BIC), are often used for this purpose. However, their application to phase-type distributions is not straightforward. The primary challenge is that PH distributions suffer from identifiability issues: different parameter sets $(\bfp, \bfT)$ of a given dimension $r$ can yield the same distribution, and a distribution of dimension $r$ may have a minimal representation of a lower dimension. This makes it difficult to determine the true number of free parameters, a key component of the penalty terms in AIC and BIC.

Given these challenges, the typical fitting process is iterative and exploratory. A common strategy is first to fit models with small dimensions and progressively increase $r$. At each step, one assesses the improvement in fit, often by analyzing changes in the loglikelihood or using visual diagnostics. Once a suitable dimension is identified, one may then explore if a more parsimonious model with a sparse structure can achieve a comparable fit, thereby reducing the number of parameters. For instance, one could consider a generalized Coxian structure where the parameters take the general form 
\begin{align*}
	\bfp = (\pi_1,...,\pi_p) \,, \quad \bfT= \left( \begin{array}{cccccc}
		-\lambda_1 &  \lambda_1 w_1 & 0 & \cdots  & 0 \\
		0 &  -\lambda_2 & \lambda_2 w_2  & \cdots  & 0 \\
		\vdots &  \vdots & \vdots & \ddots  & \vdots \\
		0 &  0 & 0 & \cdots  & -\lambda_r \\
	\end{array} \right)\,,
\end{align*}
where $\lambda_k>0$, $w_k \in [0,1]$, $k = 1,\dots, r$. Importantly, the fitting procedures previously described respect the zero-constraints within these structures, meaning that any parameter initially set to zero remains zero throughout the estimation process, ensuring that the final fitted model preserves the intended sparsity.
		
Simpler nested structures can also be considered. For example, setting $\bfp = (1, 0, \dots, 0)$ yields a standard Coxian distribution, while setting all $s_k=1$ results in a generalized Erlang distribution.

This manual process can be computationally demanding and time-consuming. Its success often depends heavily on a reasonable initial guess of the dimension $r$, which motivates the need for our proposed automated procedure described next.

\subsection{Automatic selection of $r$}
We now  introduce a data-driven procedure that automatically identifies a suitable initial dimension $r$. The core idea of the proposed method is to subsequently evaluate different dimensions and select the smallest one that adequately captures the survival pattern observed in the data. To assess adequacy, we compare fitted models against the nonparametric Kaplan-Meier (KM) estimator, which serves as a benchmark for the empirical survival experience. The procedure combines maximum-likelihood estimation via the EM algorithm with diagnostic comparisons to the KM survival curve and its confidence band.

The methodology starts by estimating the KM survival curve, $\widehat{S}_{\text{KM}}(t)$, along with its associated confidence bands, and by defining a finite, ordered sequence of candidate dimensions $r_{\min},\ldots,r_{\max}$. If preferred, a particular phase-type structure of the parameters, such as a general Coxian, can also be considered. For a given dimension $r$, a phase-type model is fitted using Algorithm~\ref{algo1}, and the resulting survival function $\widehat{S}_r(t)$ is evaluated across the grid of KM time points where $\widehat{S}_{\text{KM}}(t)<1$. To assess the model's performance, two diagnostic metrics are calculated: the {\em exception rate}, measured as the percentage of time points where $\widehat{S}_r(t)$ falls outside the confidence bands, and the {\em relative loss}, which compares the squared deviation of $\widehat{S}_r(t)$ from the KM estimate with the minimum squared deviation from the confidence band boundaries. This criterion measures whether the fitted model is, on average, closer to the KM curve itself than to the edges of its sampling uncertainty. The procedure iterates over $r$ sequentially, terminating early if a model yields no exceptions and achieves a distance to the KM curve that is smaller than its distance to the bands; otherwise, it retains the specification that minimizes the exception rate.

Since the EM algorithm is known to be sensitive to initialization and may converge to local maxima, particular care is taken in the choice of initial phase-type parameters in each iteration of $r$. For each fit, the initial probability vector  $\bfp$  and the sub-intensity matrix $\bfT$ are first randomly initialized subject to the imposed structural constraints, and subsequently adjusted using empirical survival information. Specifically, the empirical plateau of the KM estimator,
\begin{align*}
\widehat{S}_\infty = \lim_{t\to \infty} \widehat{S}_{\text{KM}}(t),
\end{align*}
is used to pre-calibrate the initial estimate of the cure fraction by assigning $\widehat{S}_\infty$ to the last component of the initial probability vector $\bfp$. The remaining probability mass ($1-\widehat{S}_\infty$) is distributed across the transient states. In addition, the sub-intensity matrix $\bfT$ is scaled so that the mean of the corresponding phase-type distribution approximately matches the empirical mean of the observed survival times. This pre-calibration yields reasonable starting values for the EM iterations and improves numerical stability and convergence behavior. The complete routine is summarized in Algorithm~\ref{alg:auto}.

\begin{algorithm}[!htbp]
\caption{Automatic selection of the dimension $r$}\label{alg:auto}
\begin{algorithmic}[1]
\Require Observations $(\xi_i,\delta_i)_{i=1,\dots,n}$, and minimum and maximum number of phases $r_{\min} < r_{\max}$.
\Ensure Dimension ${ r}$ and estimates for the  phase-type parameters $(\bfp,\bfT)$.
\State Compute the KM estimate $\widehat{S}_{\text{KM}}(t)$ and its confidence bands $[\widehat{S}_{\text{low}}(t),\,\widehat{S}_{\text{upp}}(t)]$.
\State Set $\widehat{S}_\infty=\lim_{t\to \infty}\widehat{S}_{\text{KM}}(t)$ and define the KM grid $\mathcal{T}=\{t:\widehat{S}_{\text{KM}}(t-) < \widehat{S}_{\text{KM}}(t), \, \widehat{S}_{\text{KM}}(t)<1\}$.
\State Initialize  $\text{er}_{\text{prev}}\leftarrow +\infty$.
\For{$r=r_{\min},\ldots,r_{\max}$}
    \State Initialize a phase-type model of dimension $r$ and parameters $(\bfp, \bfT)$ with $\bfT_{r,1:r}=\0$. 

    \State Set $\pi_r=\widehat{S}_\infty$ , $\bfp_{1:(r-1)} = \bfp_{1:(r-1)} (1-\widehat{S}_\infty) / (\sum_{l=1}^{r}\pi_l)$, and $\bfT = \bfT / (n^{-1}\sum_{i=1}^{n}\xi_i)$
    \State Fit the model by maximum likelihood using Algorithm~\ref{algo1}, obtaining $\widehat{S}_r(t)$.
    \State Evaluate $\widehat{S}_r(t)$ on $\mathcal{T}$ and compute:
    \begin{itemize}
        \item Exception rate $\text{er}_r$: proportion of $t\in\mathcal{T}$ such that $\widehat{S}_r(t)\notin[\widehat{S}_{\text{low}}(t),\,\widehat{S}_{\text{upp}}(t)]$.
        \item $\text{err}_2=\sum_{t\in\mathcal{T}} \big(\widehat{S}_r(t)-\widehat{S}_{\text{KM}}(t)\big)^2$.
        \item $\text{err}_1=\sum_{t\in\mathcal{T}} \min\Big\{\big(\widehat{S}_r(t)-\widehat{S}_{\text{low}}(t)\big)^2,\ \big(\widehat{S}_r(t)-\widehat{S}_{\text{upp}}(t)\big)^2\Big\}$.
    \end{itemize}
    \State \textbf{Early stop:} if $\text{er}_r=0$ and $\text{err}_2<\text{err}_1$, set $(\hat{\bfp},\hat{\bfT})\leftarrow (\bfp,\bfT)$, $\hat{r}\leftarrow r$ and \textbf{break}.
    \State \textbf{Otherwise:} if $\text{er}_r<\text{er}_{\text{prev}}$, set $(\hat{\bfp},\hat{\bfT})\leftarrow (\bfp,\bfT)$, $\hat{r}\leftarrow r$, and $\text{er}_{\text{prev}}\leftarrow \text{er}_r$.
\EndFor
\State Return the selected dimension  $\hat r$ and estimates $(\hat{\bfp},\hat{\bfT})$.
\end{algorithmic}
\end{algorithm}

This method provides an intuitive and justifiable initial estimate for $r$. From this starting point, the model can be further refined by examining the loglikelihood profile over nearby dimensions, exploring sparser nested structures, or applying the goodness-of-fit diagnostics described in the following section to ensure that the selected specification adequately captures the underlying data structure. It is important to note that this automatic selection procedure is primarily designed for settings without covariate information, where a single marginal survival curve can be meaningfully compared to the KM estimator. In the presence of covariates, survival functions become conditional on $\bfx$, and a direct KM-based comparison is no longer appropriate. Nevertheless, the procedure can still be employed to obtain a sensible initial choice of the dimension $r$ and corresponding parameter values, which may then serve as starting points for fitting the full regression model.

\subsection{Goodness-of-fit}

Once a suitable model has been selected, it is essential to assess whether it provides an adequate description of the data. In the cure-rate setting, standard goodness-of-fit tools developed for classical survival models must be interpreted with care, as they do not account for the presence of a cured subpopulation. We therefore rely on residual-based diagnostic methods specifically tailored to cure models. In particular, following the framework of~\cite{peng2017residual}, we employ Cox-Snell-type residuals to assess the overall fit of the model, together with modified Cox-Snell residuals that directly target the fit of the latency component. These diagnostics are well suited to phase-type cure models and allow for both graphical and quantitative goodness-of-fit assessment.  

\subsubsection{Cox-Snell residuals}

For a survival model with survival function $S(t)$, the Cox-Snell (CS) residuals are defined as
$c_i = -\log(S(\xi_i))$,
where $\xi_i$ denotes the observed time for subject $i$. In classical survival analysis, if the model is correctly specified, the CS residuals can be viewed as a censored sample from a standard exponential distribution with mean one ($\text{Exp}(1)$). However, 
in the cure-rate setting, this interpretation requires refinement. For our phase-type cure model, the overall survival function is
$S_\tau(t) = \boldsymbol{\pi}\exp(\mathbf{T}t)\mathbf{e}$,
and the corresponding CS residuals are
$c_i = -\log(\boldsymbol{\pi}\exp(\mathbf{T}\xi_i)\mathbf{e})$.
As shown by~\cite{peng2017residual}, even when the cure model is correctly specified, these residuals do not follow a pure exponential distribution. Instead, they form a mixed-type distribution with a continuous exponential component and a point mass at $-\log(1-p)$, corresponding to cured individuals. Importantly, this point mass is always censored, so the residuals can still be treated as a censored sample from an $\text{Exp}(1)$ distribution.

Consequently, goodness-of-fit can be assessed by computing the empirical cumulative hazard function of the residuals using the Nelson-Aalen estimator and comparing it to the identity line corresponding to the theoretical cumulative hazard of an $\text{Exp}(1)$ distribution. Any systematic deviations from this reference line serve as a clear indicator of a lack of fit in the overall model structure.

\subsubsection{Modified Cox-Snell residuals}

While standard CS residuals are useful for assessing the overall fit, they are relatively insensitive to misspecification in the latency component of a cure model. To overcome this limitation,~\cite{peng2017residual} proposed modified CS residuals that focus explicitly on the susceptible subpopulation.
In the cure rate phase-type model, the survival function conditional on being susceptible is given by Proposition~\ref{prop:cure_rate}. The modified CS residuals are then defined as
$c_i^{u} = -\log(S_u(\xi_i)).$
Under correct specification of the latency model, these residuals behave analogously to CS residuals in classical survival analysis and can be viewed as a censored sample from an $\text{Exp}(1)$ distribution.

In practice, since the cure status is not fully observed, the distribution of the modified residuals is estimated using a weighted Nelson-Aalen estimator, with weights given by the posterior probabilities of being susceptible, that is, 
\begin{align*}
w_i & =\delta_i+\left(1-\delta_i\right) \frac{p S_u(\xi_i)}{1-p+p S_u(\xi_i)}\,, \quad i = 1,\dots, n \,.
\end{align*}

The estimated cumulative hazard of the modified residuals is then compared to the identity line. Substantial departures from this reference indicate misspecification of the latency component.

To complement the graphical assessment, we also compute the Cram\'er-von Mises (CvM) distance between the empirical distribution of the modified CS residuals, $\widehat{F}_n^{CS}(t)$, and the $\text{Exp}(1)$ distribution, $F_0(t) =1 - \exp(-t)$, as proposed by~\cite{peng2017residual}. Specifically, the distance between the two distributions is computed as
\begin{align*}
D(\widehat{F}_n^{CS}(t)) = \int_{0}^{\infty}\left(\widehat{F}_n^{CS}(t) -  F_0(t)\right)^2\mathrm{d}F_0(t)
\end{align*}
and smaller values of $D(\widehat{F}_n^{CS}(t))$ indicate a closer agreement with the theoretical exponential distribution and hence a better fit of the latency model. Importantly, this CvM criteria can also be employed for the standard CS residuals. 

\begin{remark}\rm 
Although the preceding discussion is presented for a single fitted survival curve, these residual-based diagnostics extend naturally to the phase-type MoE regression setting. In the presence of covariates, CS and modified CS residuals are computed using the fitted conditional survival functions $S_\tau(t\mid\bfx)$ and $S_u(t\mid\bfx)$, respectively. The resulting residuals can still be interpreted as censored samples from an $\text{Exp}(1)$ distribution under correct model specification, allowing to apply the same graphical and quantitative diagnostics described above.
\end{remark}

\section{Simulation studies}\label{sec5}

To evaluate the performance of our proposed cure rate phase-type model, we conducted a comprehensive simulation study. We start by generating data from a mixture cure model where the true cure rate and latency distribution are known, allowing us to assess the model's ability to recover the underlying structure.

The latency component for the susceptible population is simulated from a mixture of three Gamma distributions, $\Gamma(1, 4)$, $\Gamma(4, 2)$, and $\Gamma(8, 1	)$, with corresponding mixture  weights $0.1$, $0.5$, and $0.4$. This specification yields a multi-modal density, shown in Figure~\ref{fig:pdf_mix}. The susceptible fraction is set to $p=0.8$, corresponding to a true cure rate of $0.2$. Independent censoring times are simulated from a $\mbox{Unif}(0, 30)$ distribution.  In total, $5,000$ observations are generated under this setting.

\begin{figure}[!htbp]
\centering
\includegraphics[width=0.6\textwidth, trim= 0in 0in 0in 0in,clip]{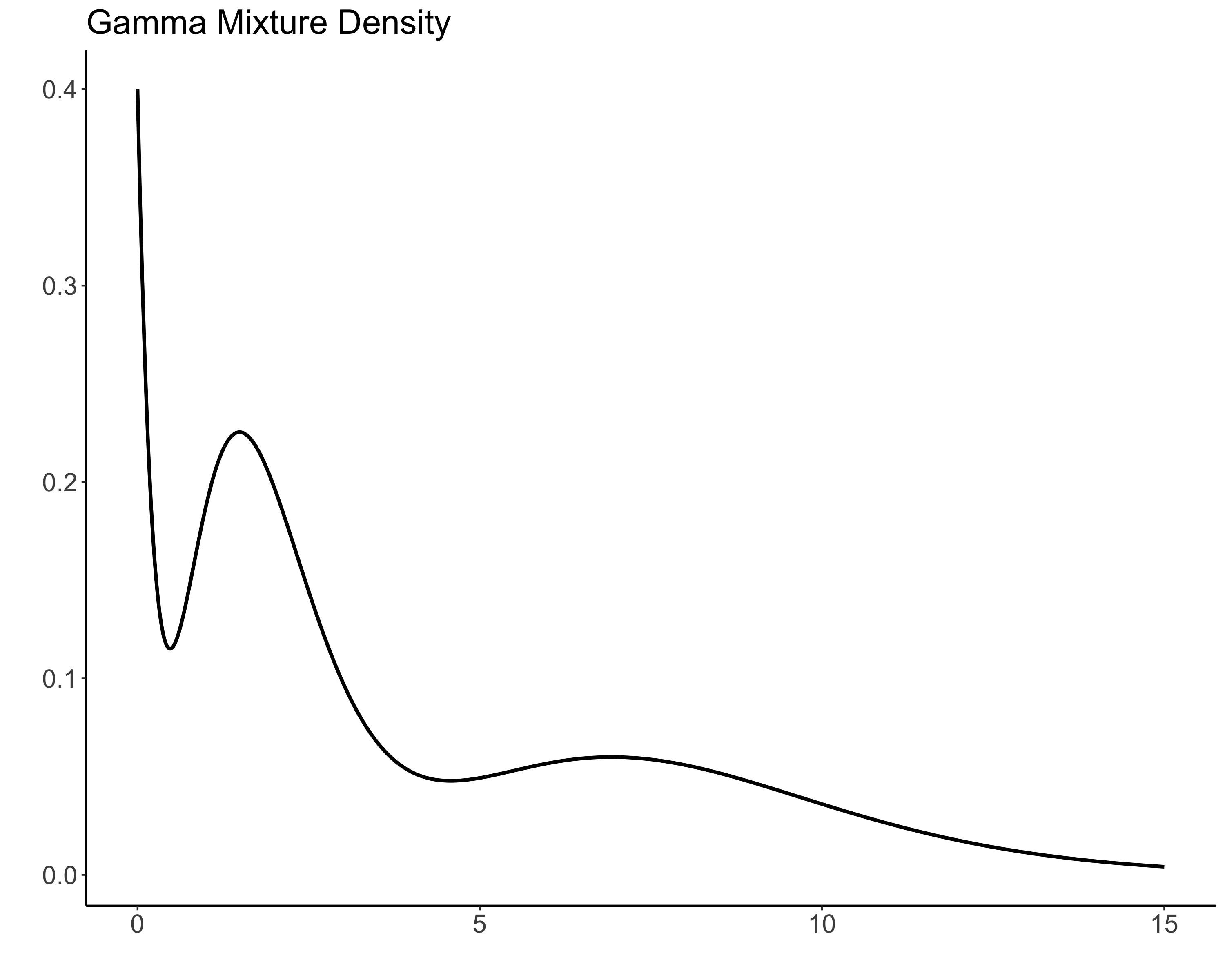}
\caption{Density function of the latency: mixture of three Gamma distributions, $\Gamma(1, 4)$, $\Gamma(4, 2)$, and $\Gamma(8, 1)$, with weights $0.1$, $0.5$, and $0.4$.}
\label{fig:pdf_mix}
\end{figure}

We then fit the proposed cure rate phase-type model to the simulated data, using a general Coxian structure of dimension $r=9$ to allow for sufficient flexibility in capturing the complex latency distribution. For comparison,  we also fitted several traditional know models, including standard parametric mixture models (Exponential and Weibull) as well as two-component mixture models (Weibull-Exponential and Weibull-Weibull),  which were estimated using the \texttt{cuRe} R package~\citep{jensen2022fitting}.

The results of the model comparison are summarized in Table~\ref{tab:chisq_statBC}. The phase-type model achieves a substantially higher loglikelihood than all other models, indicating a superior overall fit, while also providing an estimate of the susceptible fraction ($\hat{p}=0.7993$) that is remarkably close to the true value of 0.8.
This conclusion is reinforced by the goodness-of-fit diagnostics based on the CvM criterion. The improvement is particularly pronounced for the modified CS residuals, which directly assess the latency component. Here, the CvM statistic for the phase-type model (0.49) is considerable smaller than that of the next best traditional model (4.18), highlighting its unique ability to accurately capture the complex latency distribution.

\begin{table}[!htbp]
\centering
\caption{Model comparison results for the simulated dataset. The phase-type model achieves the highest loglikelihood and an accurate estimate of $p$, with substantially smaller CvM distances, particularly for the modified CS residuals.}\label{tab:chisq_statBC}
\begin{tabularx}{\textwidth}{lcccc}
  \toprule
  \textbf{Model} & $\hat{p}$ & \textbf{Loglikelihood} & \thead[l]{\textbf{CvM criterion} \\ \textbf{for CS residuals }} &  \thead[l]{\textbf{CvM criterion} \\ \textbf{for modified} \\ \textbf{CS residuals }} \\ 
  \midrule
  Exponential   &  0.8045   &  -10,338.95  &   32.00   &   7.36   \\
  Weibull   &  0.8057   &  -10,338.45  &   31.84   &   7.03   \\
  Weibull-Exponential   &  0.8010   &  -10,336.63  &   31.85   &   7.27   \\
  Weibull-Weibull   &  0.8130   &  -10,268.58  &   30.39   &   4.18   \\
  Phase-type   &  0.7993   &  -10,128.73  &   28.35   &   0.49   \\
  \bottomrule 
\end{tabularx}
\end{table}

Figure~\ref{fig:km_vs_fitted} visually confirms the adequacy of the fit of the phase-type model, with its estimated survival curve closely overlapping the KM curve from the simulated data across the entire time range, including the plateau associated with the cured fraction. In contrast, the fitted curves from the alternative models diverge from the KM estimate in some regions, particularly at early and intermediate times, highlighting their limited flexibility in capturing the complex shape of the underlying latency distribution.

\begin{figure}[!htbp]
\centering
\includegraphics[width=\textwidth, trim= 0in 0in 0in 0in,clip]{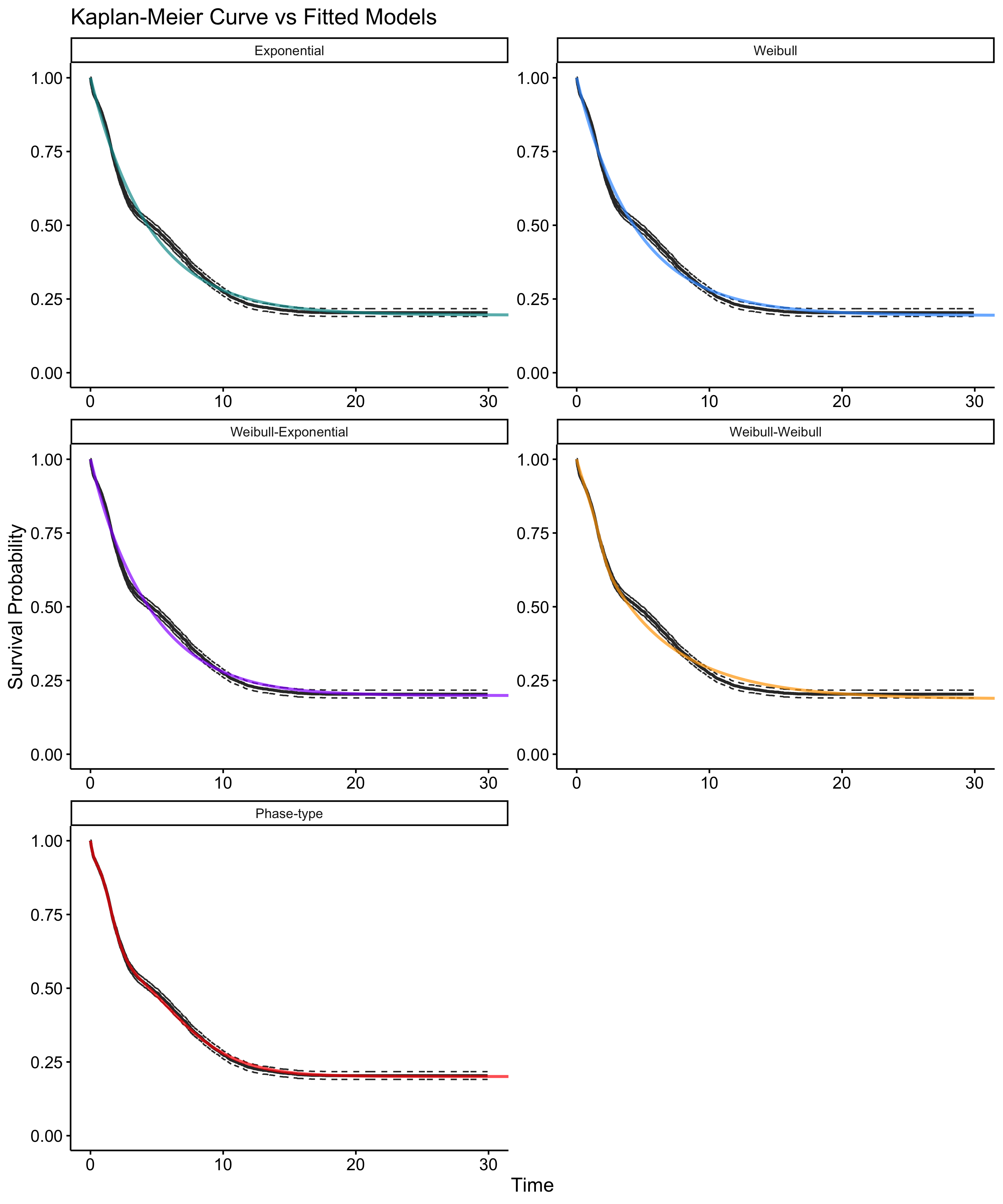}
\caption{KM estimates (black step functions with confidence bands) and fitted survival curves for the simulated dataset under each model specification. While standard models like the Exponential and Weibull show significant divergence, the phase-type model's survival curve (bottom panel) overlaps the KM curve almost perfectly, including the plateau representing the cured fraction.}
\label{fig:km_vs_fitted}
\end{figure}

To examine the effect of censoring on model performance, we varied the upper limit $U$ of the uniform censoring distribution, $\mbox{Unif}(0, U)$, from 15 to 30. Since stronger censoring is known to impair the estimation of the cure fraction, this experiment allows us to evaluate the robustness of each model under different censoring intensities. For each censoring level, we generated 100 replicated datasets and computed the average estimated susceptible fraction $\hat{p}$, the average loglikelihood, and the average CvM criteria for both the CS and modified CS residuals.  In each replication, the phase-type model was fitted using a general Coxian structure of the parameters, with the number of phases $r$ automatically selected according to the procedure described in the previous section, and restricted to a maximum of $r=10$.

Figures~\ref{fig:cure_ll} and~\ref{fig:sim_residuals} summarize the average results across the 100 replications for each censoring level. In Figure~\ref{fig:cure_ll}, we observe that as censoring weakens (larger $U$), all models show improvements in loglikelihood and cure rate estimation. However, the phase-type model consistently achieves the highest loglikelihood and the most accurate estimate of the susceptible fraction, remaining close to the true value even under strong censoring. To better visualize these differences, we also included a plot showing the loglikelihood gap between each of the other four models and the phase-type specification (top right), which always achieved the highest loglikelihood across all settings.

The residual-based CvM criteria (Figure~\ref{fig:sim_residuals}) further confirm the superior performance of the phase-type model, which maintains the smallest CvM distances for both the standard and modified CS residuals. Overall, these results highlight the robustness of the phase-type model in capturing both the cure fraction and the underlying latency distribution, even under severe censoring.

\begin{figure}[!htbp]
\centering
\includegraphics[width=0.49\textwidth, trim= 0in 0in 0in 0in,clip]{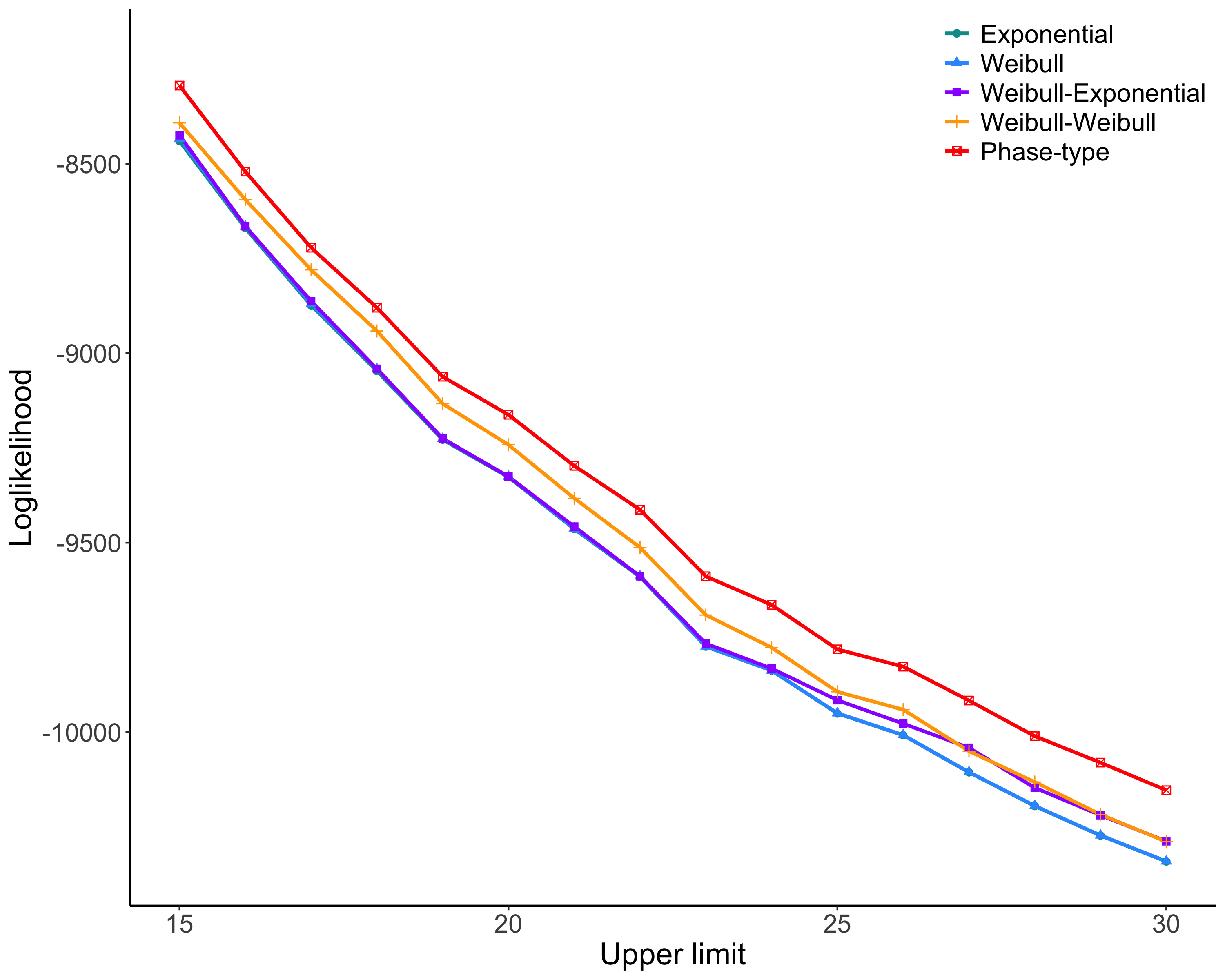}
\includegraphics[width=0.49\textwidth, trim= 0in 0in 0in 0in,clip]{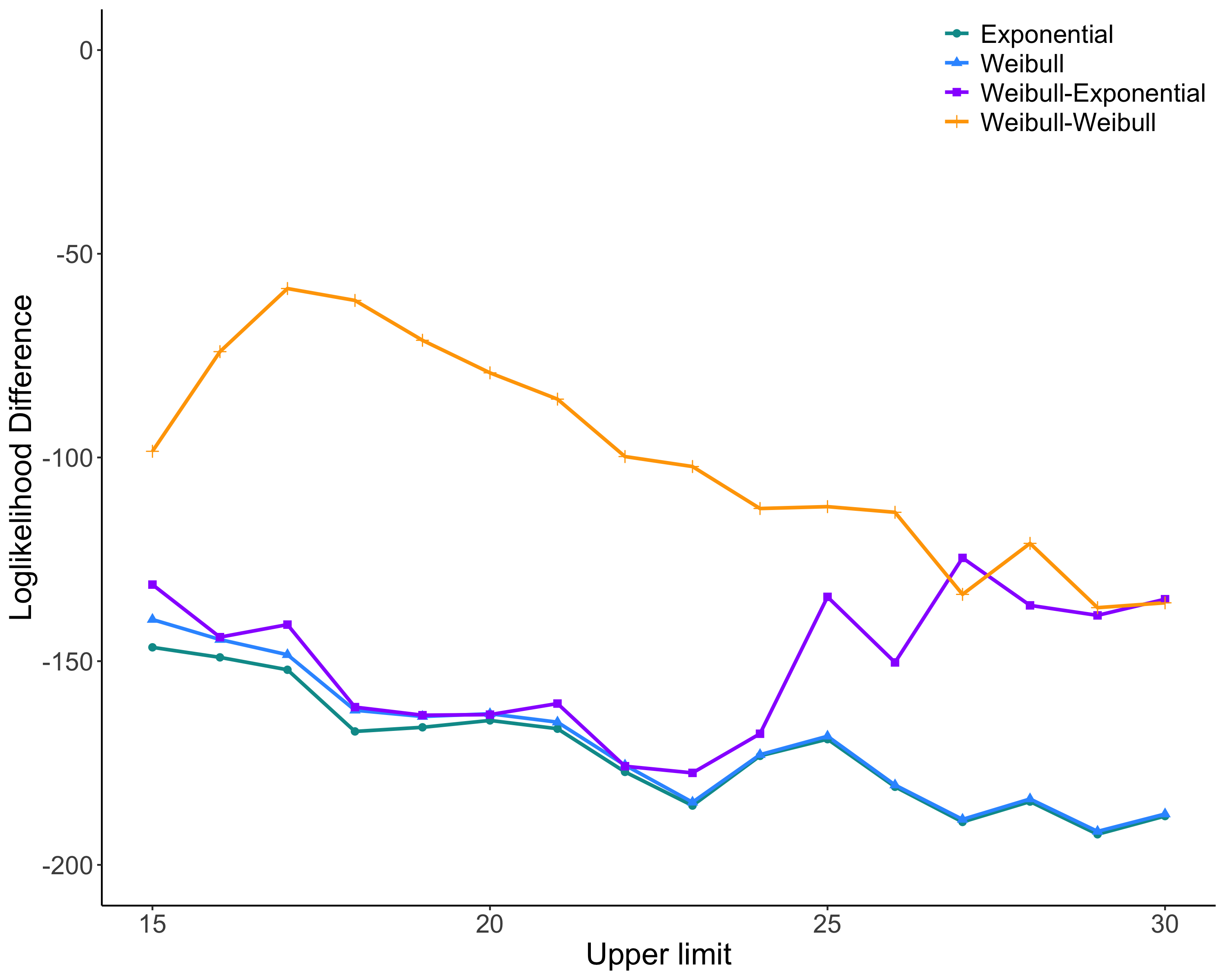}
\includegraphics[width=0.49\textwidth, trim= 0in 0in 0in 0in,clip]{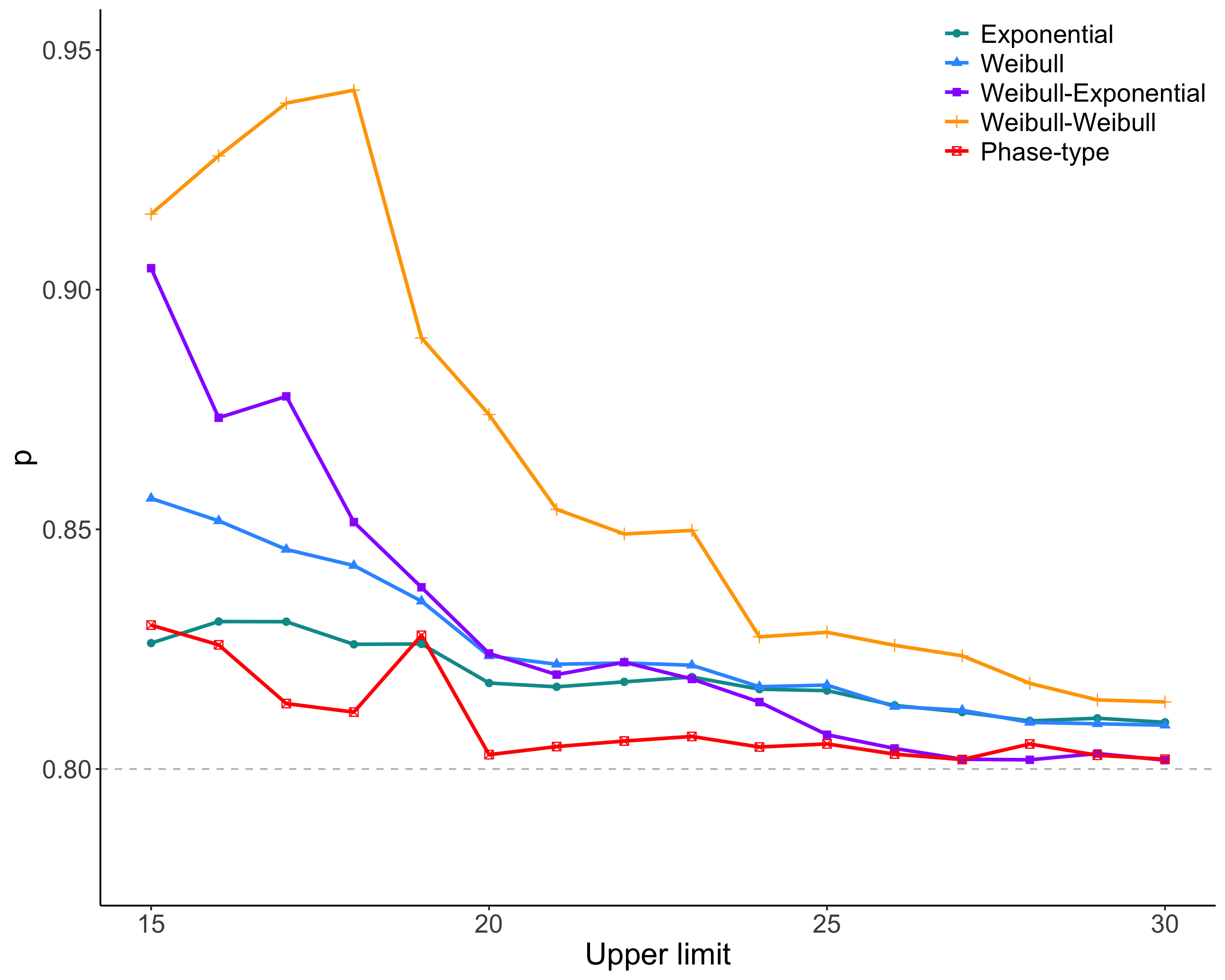}
\caption{Average models' performance across different censoring levels in 100 replications. Top left: average loglikelihoods, with the phase-type model (red) consistently highest. Top right: loglikelihood differences relative to the phase-type model, all negative. Bottom: average estimated susceptible fraction $\hat{p}$, with the true value $p=0.8$ shown as a dashed line. } 
\label{fig:cure_ll}
\end{figure}

\begin{figure}[!htbp]
\centering
\includegraphics[width=0.49\textwidth, trim= 0in 0in 0in 0in,clip]{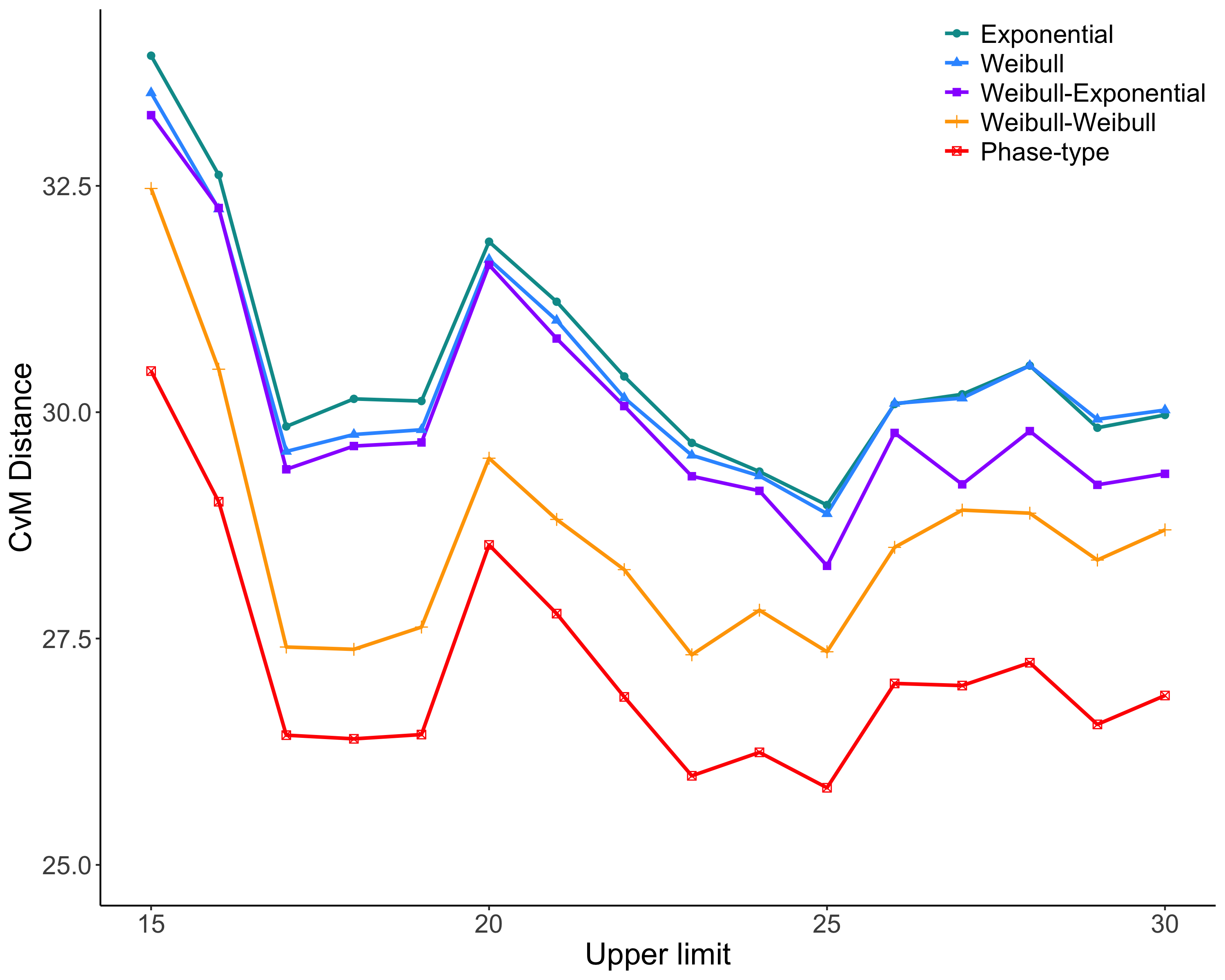}
\includegraphics[width=0.49\textwidth, trim= 0in 0in 0in 0in,clip]{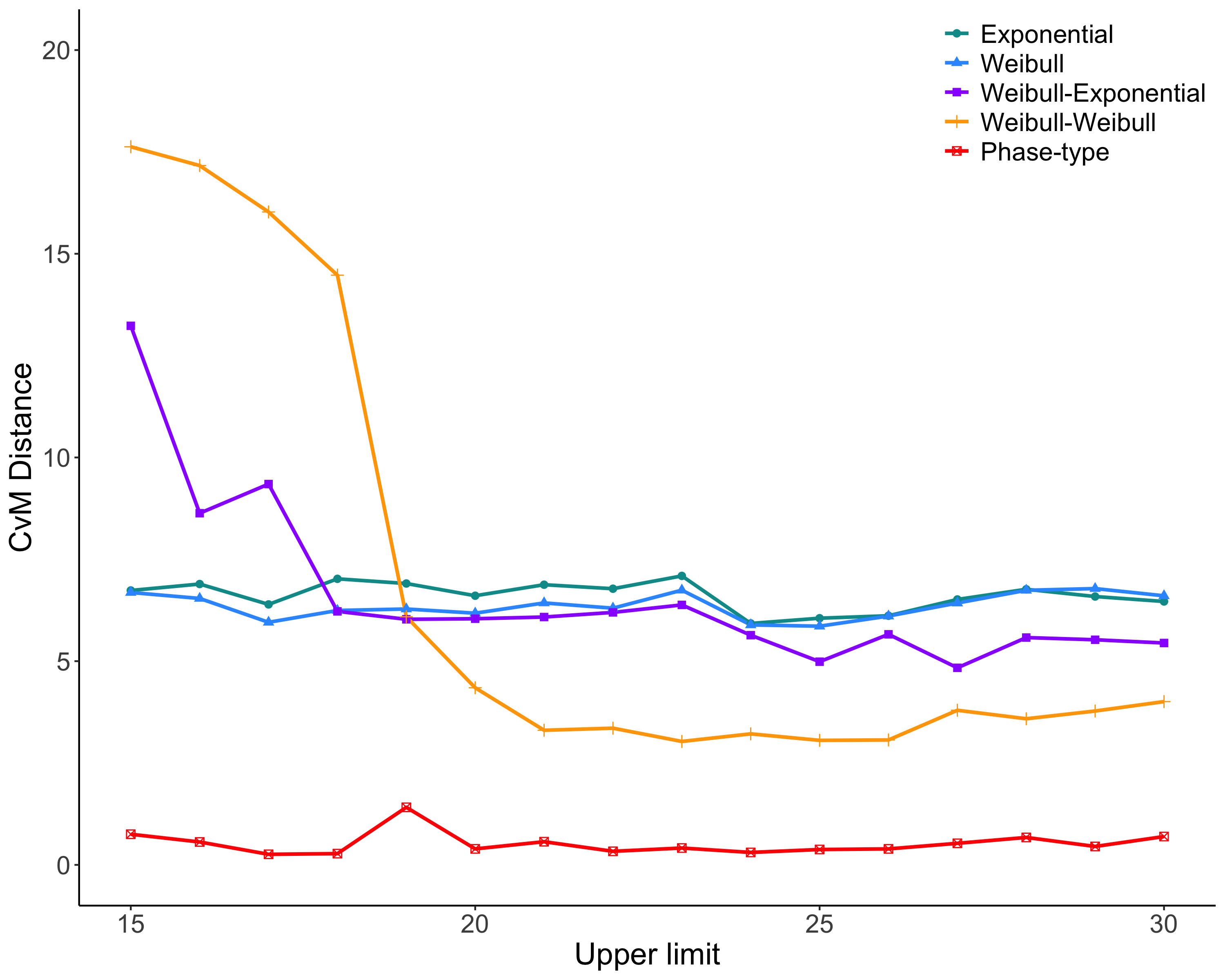}
\caption{Average Cram\'er-von Mises distances for CS (left) and modified CS residuals (right) across censoring levels. Lower values indicate a better fit. The phase-type model (red) achieves smaller distances in both cases, with a more notable difference in the modified residuals, demonstrating its superior fit to the latency distribution.} 
\label{fig:sim_residuals}
\end{figure}

\section{Leukemia data}\label{sec6}

To demonstrate the applicability of the proposed cure rate phase-type  model to real-world data, we analyzed the {\em acute lymphoblastic leukemia} dataset originally presented in~\citet{kersey1987comparison} and publicly available in, for instance, the \texttt{mixcure} R package~\citep{mixcure}. This dataset, widely used in the cure modeling literature, records the time to leukemic recurrence for patients who received either autologous or allogeneic bone marrow transplants. The data include 91 patients, followed from March 1982 to May 1987, of whom 46 received allogeneic and 45 autologous transplants. 
As several patients remained in long-term remission, the dataset exhibits a clear cured fraction, making it a benchmark example for cure models.

We model the data using a phase-type MoE specification, where the type of transplant (Allogeneic vs. Autologous) is incorporated as a categorical covariate affecting both the incidence and latency. Specifically, we fit a cure phase-type MoE model of dimension $r = 4$, general Coxian structure, and a lognormal inhomogeneity function.
Figure~\ref{fig:fit_leukemia} displays the KM survival estimates for each group along with the fitted survival curves from the phase-type MoE model, where we observe that the model successfully captures the initial steep decline in survivals as well as the subsequent plateauing of the curves. The estimated cure fractions are $1 - \hat{p}_{\text{Allogeneic}} = 0.2621$ and $1 - \hat{p}_{\text{Autologous}} = 0.1972$, indicating a slightly higher long-term remission probability among patients receiving allogeneic transplants.

\begin{figure}[!htbp]
\centering
\includegraphics[width=0.6\textwidth, trim= 0in 0in 0in 0in,clip]{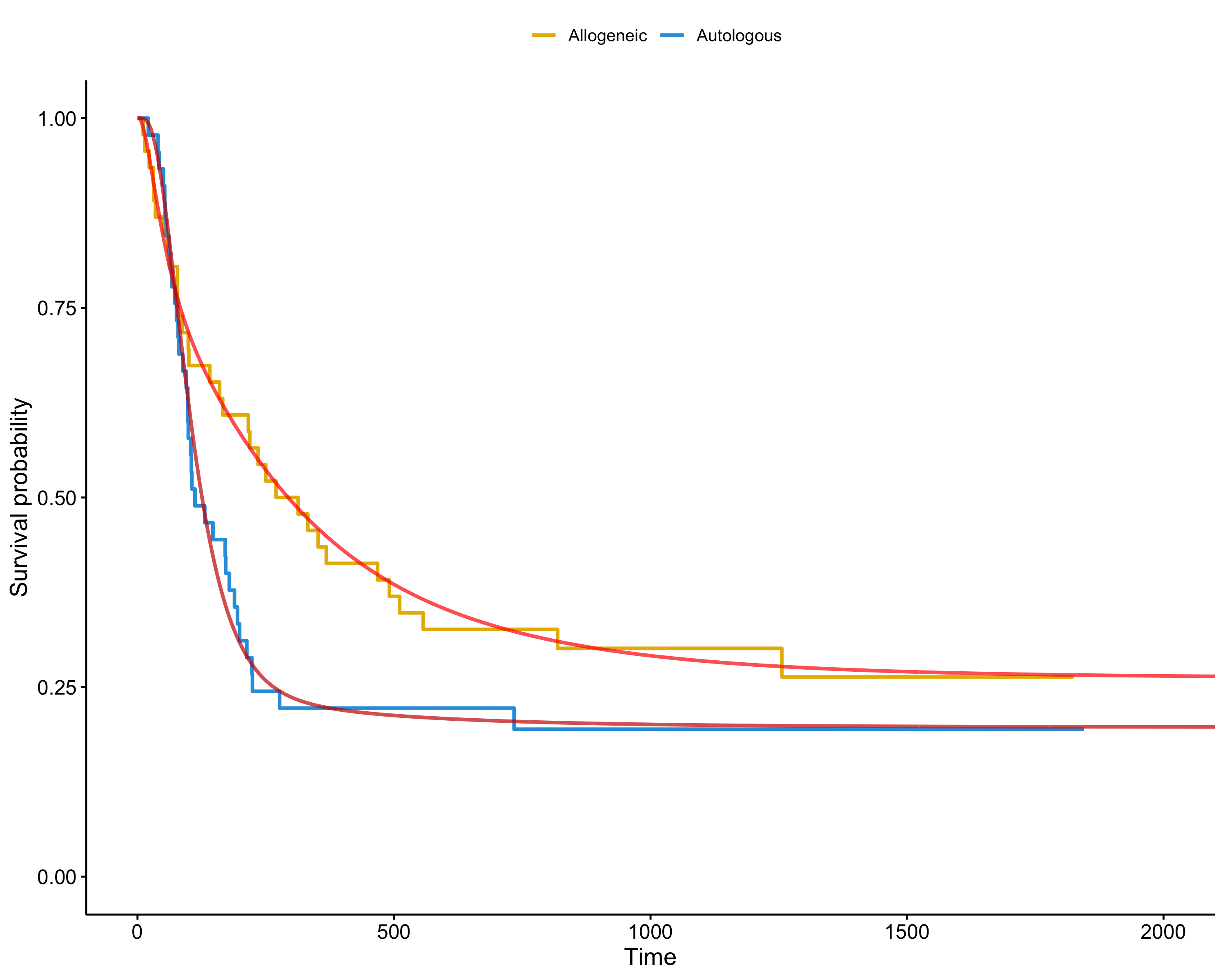}
\caption{KM estimates and fitted survival curves from the phase-type MoE model for the leukemia data, stratified by transplant type (Allogeneic in yellow, Autologous in blue). The model captures both the initial steep decline and the plateau for both populations.} 
\label{fig:fit_leukemia}
\end{figure}

To formally assess the goodness-of-fit, we analyzed the residuals of the model. Figure~\ref{fig:residuals_leukemia} shows the empirical cumulative hazards estimated via the Nelson-Aalen estimator of the standard and modified CS residuals. The close alignment of the empirical curves with the identity line, provides strong evidence that the residuals follow an $\mbox{Exp}(1)$ distributions, indicating that the model is well-specified. We further quantified the quality of the fit using the CvM criterion computed on the modified residuals. The cure phase-type MoE model achieves a value of 3.6, outperforming the mixture models considered in~\citet{kersey1987comparison}, namely the Exponential (21.0), Weibull (8.9), and Lognormal (4.9) specifications.

\begin{figure}[!htbp]
\centering
\includegraphics[width=0.49\textwidth, trim= 0in 0in 0in 0in,clip]{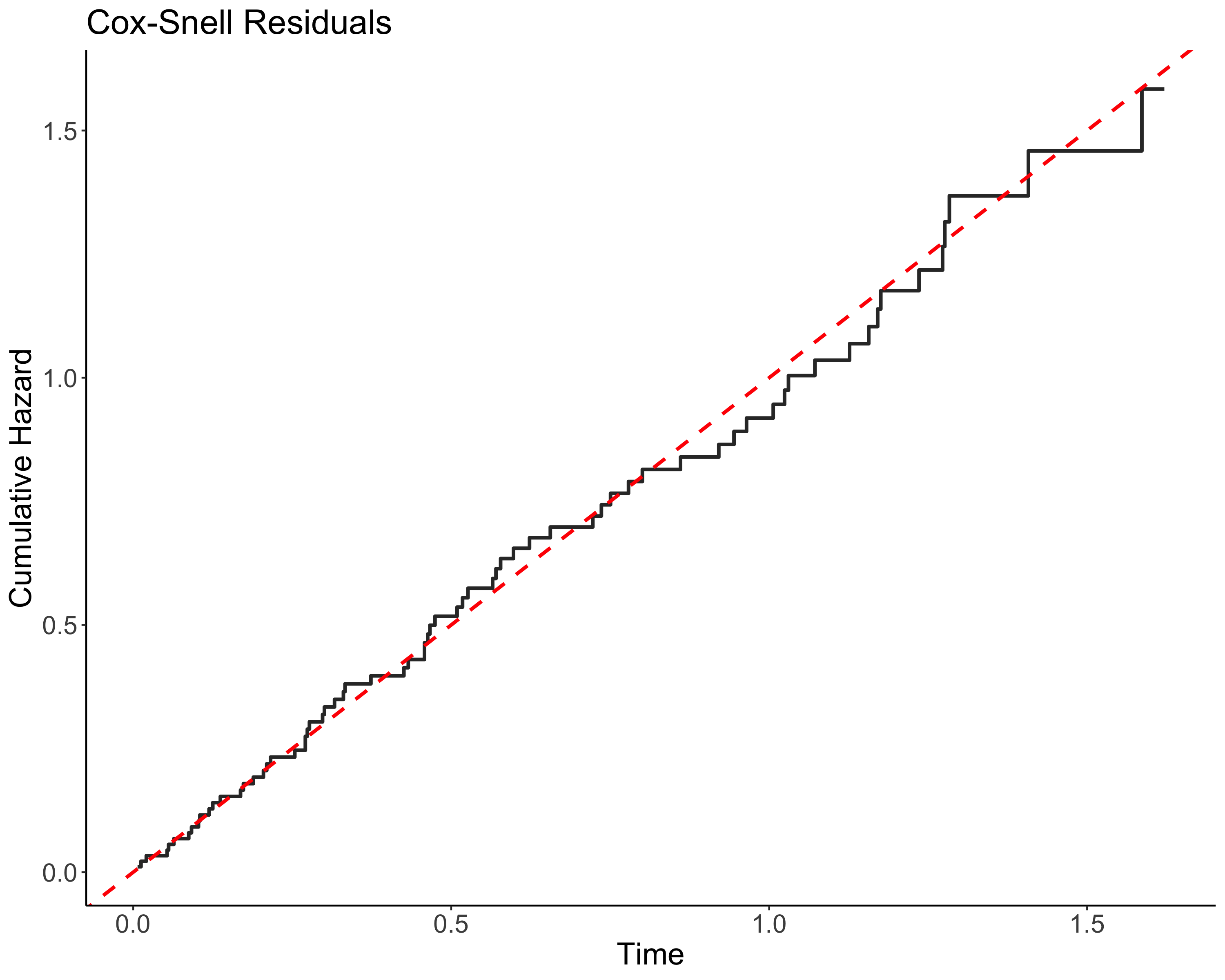}
\includegraphics[width=0.49\textwidth, trim= 0in 0in 0in 0in,clip]{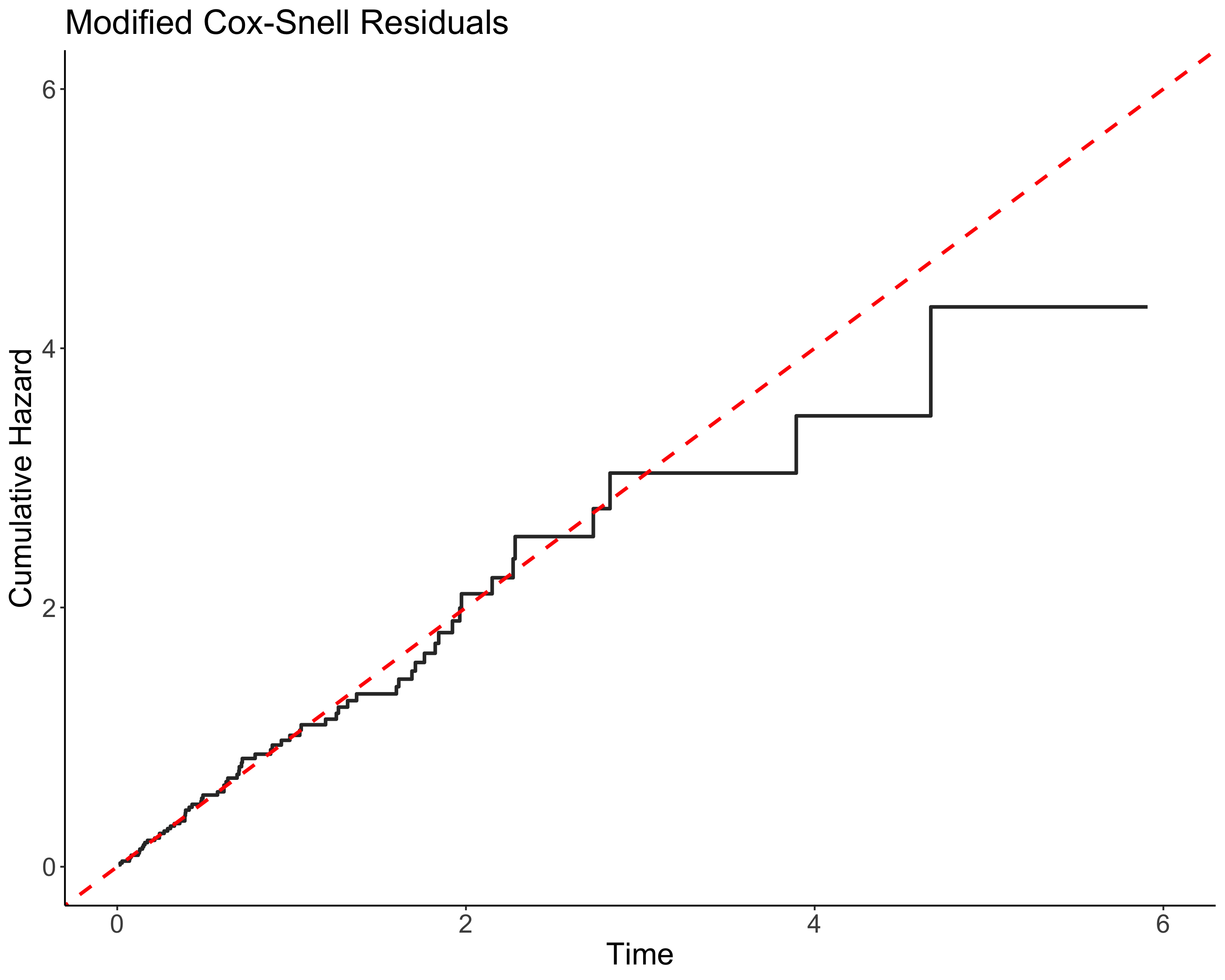}
\caption{Nelson-Aalen estimates of the cumulative hazard for CS (left) and modified CS (right) residuals, compared to the $\mbox{Exp}(1)$ reference line (dashed). The close alignment in both panels indicates adequate model fit.} 
\label{fig:residuals_leukemia}
\end{figure}

\section{Conclusion} \label{sec:conc}

In this paper, we have introduced a novel class of cure rate models based on phase-type distributions. By defining a Markov jump process with two absorbing states, one representing the event of interest and another representing immunity, our framework provides a dynamic interpretation of the curing mechanism that generalizes the static mixture cure model. The proposed approach offers several advantages, including explicit closed-form expressions for the cure rate and the latency distribution, and a flexible regression framework via the Mixture-of-Experts that can approximate any other cure rate model. 

Our numerical experiments demonstrated the practical benefits of this methodology. In simulations, the phase-type specification achieved substantially higher loglikelihoods and more accurate cure rate estimates than standard parametric alternatives, with goodness-of-fit diagnostics confirming its ability to capture complex, multimodal latency distributions even under strong censoring. The application to the leukemia dataset further illustrated its utility in real-world settings.

Several directions for future research emerge from this work. An important and natural extension includes multivariate cure models using multivariate phase-type distributions.

\subsection*{Funding}
The research of M.B. was supported by the Carlsberg Foundation, grant CF23-1096.

\subsection*{Competing interests}
No competing interest is declared.

\subsection*{Author contributions statement}

M.B. and J.Y. contributed equally to both the theory and implementation. 

\newpage
\bibliography{curePH_new.bib}

\end{document}